\documentclass[a4paper, 11pt]{article}

\usepackage[margin=2cm]{geometry}
\usepackage{amsmath, amsthm, amssymb, enumitem, multirow, bbm, array, url}
\usepackage{tikz}

\theoremstyle{plain}

\newtheorem{corollary}{Corollary}
\newtheorem{lemma}{Lemma}
\newtheorem{proposition}{Proposition}
\newtheorem{theorem}{Theorem}

\theoremstyle{definition}

\newtheorem{problem}{Problem}

\newtheorem{claim}{Claim}

\newcommand{\IG}{\mathbb{D}}

\newcommand{\Fix}{\mathrm{Fix}}
\newcommand{\fix}{\mathrm{fix}}

\newcommand{\Per}{\mathrm{Per}}
\newcommand{\per}{\mathrm{per}}

\newcommand{\Ima}{\mathrm{Ima}}

\newcommand{\functions}{\mathrm{F}}
\newcommand{\feedback}{\tau}
\newcommand{\packing}{\nu}
\newcommand{\guessing}{\mathrm{g}}

\newcommand{\neighbourhood}{N}
\newcommand{\inn}[1]{#1_\mathrm{in}}
\newcommand{\NIn}{\inn{\neighbourhood}}
\newcommand{\out}[1]{#1_\mathrm{out}}
\newcommand{\NOut}{\out{\neighbourhood}}

\newcommand{\dIn}{\inn{d}}

\newcommand{\girth}{\gamma}

\newcommand{\rank}{\mathrm{rank}}
\newcommand{\avg}{\mathrm{avg}}

\newcommand{\minrank}{\rank^-}
\newcommand{\avgrank}{\rank}
\newcommand{\maxrank}{\rank^+}

\newcommand{\crank}{\rank^\land}

\newcommand{\maxper}{\per^+}

\newcommand{\minfix}{\fix^-}

\newcommand{\maxfix}{\fix^+}

\renewcommand{\r}{\mathrm{r}}

\newcommand{\cliqueCover}{\mathrm{\pi}}

\newcommand{\fractional}[1]{#1^*}
\newcommand{\fractionalPacking}{\fractional{\packing}}

\newcommand{\fractionalCliqueCover}{\fractional{\cliqueCover}}

\newcommand{\loopfull}[1]{{#1}^\circ}

\DeclareFontFamily{U}{matha}{}
\DeclareFontShape{U}{matha}{m}{n}{
  <-5.5>    matha5
  <5.5-6.5> matha6 
  <6.5-7.5> matha7
  <7.5-8.5> matha8
  <8.5-9.5> matha9
  <9.5-11>  matha10
  <11->     matha12
}{}
\DeclareSymbolFont{matha}{U}{matha}{m}{n}
\DeclareFontSubstitution{U}{matha}{m}{n}
\DeclareFontFamily{U}{mathx}{\hyphenchar\font45}
\DeclareFontShape{U}{mathx}{m}{n}{<-> mathx10}{}
\DeclareSymbolFont{mathx}{U}{mathx}{m}{n}
\DeclareFontSubstitution{U}{mathx}{m}{n}

\DeclareMathDelimiter{\ldbrack}{4}{matha}{"76}{mathx}{"30}
\DeclareMathDelimiter{\rdbrack}{5}{matha}{"77}{mathx}{"38}
\DeclareMathSymbol{\bigovoid}{1}{mathx}{"EC}

\newcommand{\<}{\ldbrack}
\renewcommand{\>}{\rdbrack}

\title{On the influence of the interaction graph on a finite dynamical system}
\author{Maximilien Gadouleau\footnote{Departement of Computer Science, Durham University, Durham, UK. m.r.gadouleau@durham.ac.uk}}
\date{\today}

\begin{document}

\maketitle

\begin{abstract}
A finite dynamical system (FDS) is a system of multivariate functions over a finite alphabet, that is typically used to model a network of interacting entities. The main feature of a finite dynamical system is its interaction graph, which indicates which local functions depend on which variables; the interaction graph is a qualitative representation of the interactions amongst entities on the network. As such, a major problem is to determine the effect of the interaction graph on the dynamics of the FDS. In this paper, we are interested in three main properties of an FDS: the number of images (the so-called rank), the number of periodic points (the so-called periodic rank) and the number of fixed points. In particular, we investigate the minimum, average, and maximum number of images (or periodic points,  or fixed points) of FDSs with a prescribed interaction graph and a given alphabet size; thus yielding nine quantities to study. The paper is split into two parts. The first part considers the minimum rank, for which we derive the first meaningful results known so far. In particular, we show that the minimum rank decreases with the alphabet size, thus yielding the definition of an absolute minimum rank. We obtain lower and upper bounds on this absolute minimum rank, and we give classification results for graphs with very low (or highest) rank. The second part is a comprehensive survey of the results obtained on the nine quantities described above. We not only give a review of known results, but we also give a list of relevant open questions.
\end{abstract}

\maketitle

\section{Introduction}

Networks of interacting entities can be modelled as follows. The network consists of $n$ entities, where each entity $v$ has a local state represented by a $q$-ary variable $x_v \in \<q\> = \{0,1,\dots,q-1\}$, which evolves according to a deterministic function $f_v : \<q\>^n \to \<q\>$ of all the local states. More concisely, the state of the system is $x = (x_1,\dots, x_n) \in \<q\>^n$, which evolves according to a deterministic function $f = (f_1,\dots,f_n) : \<q\>^n \to \<q\>^n$, referred to as a Finite Dynamical System (FDS). FDSs have been used to model different networks, such as gene networks, neural networks, social networks, or network coding (see \cite{GR16} and references therein for the applications of FDSs). In view of their versatility, a recent stream of work is devoted to the study of FDSs per se, without a particular application in mind.

%FDSs have been used to model gene networks (see \cite{KS08,TD90}), neural networks \cite{ADG04,Hop82}, network coding \cite{Rii07}, social interactions \cite{GT83,PS83} and more (see \cite{GM90}). 

The architecture of an FDS $f: \<q\>^n \to \<q\>^n$ can be represented via its \textbf{interaction graph} $\IG(f)$, which indicates which update functions depend on which variables.  A major problem about FDSs is then to predict some of their dynamical features according to their interaction graphs. Perhaps the first example of such a result is due to Robert \cite{Rob80}, who showed that if the interaction graph is acyclic, then $f^n(x) = c$ for some $c \in \<q\>^n$. However, due to the wide variety of possible local functions, determining properties of an FDS given its interaction graph is in general a difficult problem.

In this paper, we are interested in the following three dynamical features of an FDS. A \textbf{fixed point} of $f$ is a stationary state (i.e. $f(x) = x$); a \textbf{periodic point} is a recurring state (i.e. $f^k(x) =x$ for some $k$); and an \textbf{image} is a reachable state (i.e. $f(y) = x$ for some $y$). We consider the number of fixed points, of periodic points, and of images of FDSs. In order to illustrate the influence of the interaction graph, we consider the set of all FDSs $f : \<q\>^n \to \<q\>^n$ with a given interaction graph $D$, and we study the minimum, average and maximum value of the three properties described above. This yields nine quantities for a given graph $D$ and alphabet size $q$.

Those nine quantities have been studied to various degrees. Arguably the quantity that attracted the most attention is the maximum number of fixed points, in particular due to its relationship with network coding \cite{Rii06, Rii07, GR11}. On the other hand, two quantities have not been studied so far. Firstly, the average number of periodic points seems very complex to estimate, as even in the case where $n = 1$ it requires some sophisticated machinery (see below for more explanation). Secondly, the first part of this paper is devoted to the study of the minimum number of images, where we can obtain some interesting results.

We decided to limit the scope of this survey and to present a comprehensive survey for that scope, instead of a limited survey of a broader topic. Nonetheless, let us mention three main strands of work on the dynamics of FDSs that do not fit the scope of this paper. Firstly, the signed interaction graph of an FDS not only encodes the fact that $f_v$ depends on the variable $x_u$, but also whether it is a monotonically non-decreasing (arc signed positively) or monotonically non-increasing (arc signed negatively) or neither (both signs given on the arc) function of $x_u$. A large amount of work considers the influence of the signed interaction graph (see \cite{PR10} for a survey of such work); we shall use some of these results but we will not carry out a comprehensive survey of the signed interaction graph, let alone of its local variants. Secondly, some work also restricts the nature of the local functions $f_v$: it can be linear \cite{GR11, GRF16}, threshold \cite{GT83, Gol85}, conjunctive or disjunctive \cite{ADG04a, ARS14}, etc. Thirdly, a large body of work is devoted to the study of different update schedules, where instead of applying $f$ to $x$, for instance only one entity updates its state at each time step, and $x$ becomes $(x_1, \dots, x_{i-1}, f_i(x), x_{i+1}, \dots, x_n)$ for some $i$. The same FDS may yield completely different dynamics under different update schedules (see \cite{DNS12, GM12, NS17, BCG18} for instance).

The rest of the paper is organised as follows. Section \ref{sec:background} gives some formal definitions of FDSs and related concepts. Section \ref{sec:min_rank} then is devoted to the study of the minimum number of images. Section \ref{sec:survey} gives a survey of the nine properties described above, with a review of known results and a list of open problems.

\section{Formal definitions} \label{sec:background}

A (directed) graph is a pair $D = (V, E)$, where $V$ is the set of vertices and $E \subseteq V^2$ is the set of arcs. For a comprehensive account of graphs, the reader is directed to \cite{BG09a}. We shall use the following terminology and notation. The in-neighbourhood of a vertex $v$ is $\NIn(v) = \{u : uv \in E \}$ and its in-degree is $\dIn(v) = |\NIn(v)|$; out-neighbourhoods and out-degrees are defined similarly. The in-neighbourhood of a subset of vertices $S$ is $\NIn(S) = \bigcup_{s \in S} \NIn(s)$. A vertex with an empty in-neighbourhood is a source, while a vertex with an empty out-neighbourhood is a sink. All paths and cycles are directed unless stated otherwise. A loop on the vertex $v$ is the arc $vv$. The girth of $D$, denoted $\girth(D)$, is the shortest length of a cycle in $D$, or is infinity if $D$ is acyclic. A feedback vertex set is a set of vertices $I$ such that $D[V \setminus I]$ is acyclic; the minimum size of a feedback vertex set of $D$ is the transversal number of $D$, denoted as $\feedback(D)$. More relevant concepts will be given as required.

Let $n$ be a positive integer and denote $[n] := \{1, \dots, n\}$. Let $q$ be an integer greater than or equal to $2$ and denote $\<q\> := \{0, 1, \dots, q-1\}$. A state is any $x = (x_1, \dots, x_n) \in \<q\>^n$, where $x_i \in \<q\>$ is a local state. We shall use the following shorthand notation. First of all, we identify an element $i$ of $[n]$ and the corresponding singleton $\{i\}$. Also, for any set $S = \{s_1, \dots, s_k\} \subseteq [n]$, we denote $x_S = (x_{s_1}, \dots, x_{s_k})$, only taking the order of elements into account when necessary. Moreover, we denote the set $[n] \setminus S$ as $-S$, thus denoting $x = (x_i, x_{-i})$.

We denote the set of functions $f : \<q\>^n \to \<q\>^n$ as $\functions(n,q)$. A Finite Dynamical System (FDS) is any function $f \in \functions(n,q)$. In particular, a Boolean network is any FDS in $\functions(n,2)$. We view $f$ as $f = (f_1, \dots, f_n)$, where each $f_i : \<q\>^n \to \<q\>$ is a local function of the system. We use the same shorthand notation as above for functions as well, e.g. $f_S = (f_{s_1}, \dots, f_{s_k}) : \<q\>^n \to \<q\>^k$. 

The interaction graph of $f$, denoted $\IG(f)$, has vertex set $[n]$ and $uv$ is an arc in $\IG(f)$ if and only if $f_v$ depends essentially on $x_u$, i.e.
\[
	\exists a, b \in \<q\>^n \text{ such that } a_{-u} = b_{-u}, f_v(a) \ne f_v(b).
\]
Let $D = ([n], E)$ be a graph; we have $E \subseteq [n]^2$, i.e. $D$ is a directed graph, possibly with loops. Then we consider two sets of functions:
\begin{align*}
	\functions[D,q] &:= \{ f \in \functions(n,q) : \IG(f) = D  \},\\
	\functions(D,q) &:= \{ f \in \functions(n,q) : \IG(f) \subseteq D \}.
\end{align*}

We consider the following three dynamical properties.
\begin{enumerate}
	\item An image of and FDS $f$ is simply $x \in \<q\>^n$ such that there exists $y \in \<q\>^n$ with $x = f(y)$. The set of images is denoted $\Ima(f)$ and its size is the rank of $f$: $\rank(f) = |\Ima(f)|$.

	\item A periodic point is $x \in \<q\>^n$ such that there exists $k \ge 1$ with $f^k(x) = x$. The set of periodic points is denoted $\Per(f)$ and its size is the periodic rank of $f$: $\per(f) = |\Per(f)|$.

	\item A fixed point is $x \in \<q\>^n$ such that $f(x) = x$. The set of fixed points is denoted $\Fix(f)$ and its size is denoted $\fix(f) = |\Fix(f)|$.
\end{enumerate}

%Let $f : \{0,1\}^3 \to \{0,1\}^3$ as follows.\\ (Give its interaction graph?)
%~\\
%\begin{center}
%\begin{tikzpicture}
%	\node (0) at (6,1) {000};
%	\node (1) at (0,3) {001};
%	\node (2) at (1,2) {010};
%	\node (3) at (6,3) {011};
%	\node (4) at (2,3) {100};
%	\node (5) at (1,0) {101};
%	\node (6) at (5,0) {110};
%	\node (7) at (7,0) {111};
%	
%	\draw[-latex] (1) -- (2);
%	\draw[-latex] (4) -- (2);
%	\draw[-latex] (2) -- (5);
%	
%	%loop on 5
%	
%	\draw[-latex] (3) -- (0);
%	\draw[-latex] (0) -- (7);
%	\draw[-latex] (7) -- (6);
%	\draw[-latex] (6) -- (0);
%\end{tikzpicture}
%\end{center}
%~\\
%Fixed points: $\Fix(f) = \{101\}$\\
%~\\
%Periodic points: $\Per(f) = \{101, 000, 111, 110\}$\\
%~\\
%Images: $\Ima(f) = \{101, 000, 111, 110, 010\}$

We are interested in the following nine quantities: the minimum, average, and maximum rank, periodic rank and number of fixed points of a function in $\functions[D,q]$. We denote
\begin{align*}
	\minrank[D,q] &:= \min \{ \rank(f) : f \in \functions[D,q] \}\\
	\avgrank[D,q] &:= \avg \{ \rank(f) : f \in \functions[D,q] \}\\
	\maxrank[D,q] &:= \max \{ \rank(f) : f \in \functions[D,q] \}.
\end{align*}
We use similar notation for the periodic rank and the number of fixed points. Moreover, we also consider their counterparts in $\functions(D,q)$; again, similar notation is used.

\section{Minimum rank} \label{sec:min_rank}

\subsection{Preliminary results}

For a given graph $D$, the minimum rank of a function in $\functions[D,q]$, viewed as a function of $q$, is particularly well-behaved.

\begin{lemma} \label{lem:mr_decreasing}
For any $D$, $\minrank[D,q]$ is a non-increasing function of $q$.
\end{lemma}

\begin{proof}
Let $f \in \functions[D,q]$, then consider $\hat{f} \in \functions(n, q+1)$ defined as $\hat{f}(x) = f(\hat{x})$, where $\hat{x}_v = \min( x_v, q-1 )$ for all $v$. Then it is easily checked that $\IG(\hat{f}) = D$ and $\rank(\hat{f}) \le \rank(f)$. 
\end{proof}

We can then consider the overall minimum rank
\[
	\minrank[D] := \min \{ \minrank[D,q] : q \ge 2 \}.
\]
In fact, this quantity remains unchanged if we consider infinite alphabets as well; as such, $\minrank[D]$ is a fundamental property of a graph.

\begin{corollary} \label{cor:mr}
For any $D$,
\[
	\minrank[D] = \lim_{q \to \infty} \minrank[D, q] \le \minrank[D, 2] \le 2^n.
\]
\end{corollary}

We refine Corollary \ref{cor:mr} by giving an upper bound on the smallest $q$ such that $\minrank[D,q] = \minrank[D]$. In particular, this shows that $\minrank[D]$ is computable.

\begin{theorem}
For all $D$ with $n$ vertices and $m$ arcs, $\minrank[D] = \minrank[D, (n+1)m]$.
\end{theorem}

\begin{proof}
Let $D = ([n],E)$ and $f \in \functions[D, q]$; we shall construct $g \in \functions[D, Q]$ with $Q \le (n+1)m$ and $\rank(g) \le \rank(f)$.

For any arc $ij \in E$, let $a := a_f^{ij}, b := b_f^{ij} \in \<q\>^n$ illustrate the influence of $x_i$ on $f_j$, i.e.
\[
	a_i \ne b_i, \qquad a_{-i} = b_{-i}, \qquad f_j(a) \ne f_j(b). 
\]
Let $X_f \subseteq \<q\>$ be the set of all values that appear in all the $a^{ij}, b^{ij}$:
\begin{align*}
	X_f &:= \{ a^{ij}_k :  k \in [n], ij \in E \} \cup \{ b^{ij}_k :  k \in [n], ij \in E \}\\
	&:= \{ a^{ij}_k :  k \in [n], ij \in E \} \cup \{ b^{ij}_i :  ij \in E \}.
\end{align*}
Then $|X_f| \le (n + 1)m$. Similarly, let $Y_f$ be the set of values the $f_j$ functions take on those states:
\[
	Y_f := \{ f_j(a^{ij}) : ij \in E \} \cup \{ f_j(b^{ij}) : ij \in E \}.
\]
Then $|Y_f| \le 2m$. Let $A$ be the larger of these two sets of values, i.e. $A_f = X$ if $|X_f| \ge |Y_f|$ and $A_f = Y_f$ otherwise. Thus $Q := |A_f| \le (n+1)m$. 

By translation and conjugation, we can assume that $f$ is such that $X_f, Y_f \subseteq A_f = \<Q\>$.  More explicitly, we prove the following claim.

\begin{claim}
There exists $h \in \functions[D, q]$ such that
\begin{enumerate}
	\item for all $ij \in E$, $h_j( a_h^{ij} ) \ne h_j( b_h^{ij} )$,

	\item $X_h, Y_h \subseteq A_h = \<Q\>$,

	\item $\rank(h) = \rank(f)$.
\end{enumerate}
\end{claim}

\begin{proof}
Suppose that $|X_f| \ge |Y_f|$, so that $A_f = X_f$; the proof is similar in the other case. Let $\rho, \sigma$ be two permutations of $\<q\>$ such that $\pi( Y_f ) \subseteq X_f$ and $\sigma( X_f ) = \<Q\>$. We let $\pi$ and $\sigma$ act on $\<q\>^n$ componentwise. Then let $h \in \functions[D,q]$ be defined as
\[
	h = \sigma \circ \rho \circ f \circ \sigma^{-1},
\]
and let $a^{ij}_h = \sigma( a^{ij}_f )$ and $b^{ij}_h = \sigma( b^{ij}_f )$. It is easy to verify that $h$ satisfies the claim.
\end{proof}

Let $g \in \functions(n,Q)$ be defined as follows. For all $v \in [n]$ and all $x \in \<Q\>^n$,
\[
	g_v(x) = \min ( h_v(x), Q-1 ).
\]
Since $g_j(a^{ij}) = h_j( a^{ij} )$ and $g_j(b^{ij}) = h_j(b^{ij})$ for all $ij \in E$, we obtain that the interaction graph of $g$ is $D$. Moreover, the rank of $g$ is clearly no more than that of $f$. 
\end{proof}

\subsection{Canonical interaction graphs}

In this section, we show that the minimum rank $\minrank[D,q]$ can be studied by converting the graph $D$ into a canonical form. First of all, it is clear that if $D$ is not connected, say its connected components are $D_1, \dots, D_k$, then $\minrank[D,q] = \prod_{i=1}^k \minrank[D_i,q]$. Thus, let $D = (V, E)$ be a connected graph. The canonical version of $D$, denoted as $\mathrm{C}(D)$, is obtained as follows.
\begin{enumerate}
	\item Let $V' = V_0 \cup V_1$, where $V_0 = \{v_0 : v \in V \}$ and $V_1 := \{ v_1 : v \in V \}$ are copies of $V$. Let $E' = \{ u_0v_1 : uv \in E \}$. Then let $D' = (V', E')$.

	\item Say a vertex $v_1 \in V_1$ of $D'$ is redundant if there exists $S \subseteq V_1 \setminus v_1$ such that $\NIn(S) = \NIn(v_1)$ and that if $S$ contains exactly one non-isolated vertex $u_1$, then $u < v$. Then let $R_1$ be the set of redundant vertices of $D'$ and $D'' = D' \setminus R_1$.

	\item Say a vertex $v_0 \in V_0$ of $D''$ is redundant if either $v_0$ is isolated, or there exists $u < v$ such that $\NOut(u_0) = \NOut(v_0)$. Then let $R_0$ be the set of redundant vertices of $D''$ and $\mathrm{C}(D) = D'' \setminus R_0$.
\end{enumerate}

\begin{lemma} \label{lem:canonical}
For any $D$ and $q$, we have $\minrank[D,q] = \minrank[\mathrm{C}(D), q]$.
\end{lemma}

\begin{proof}
The proof follows the three steps of the construction of $C := \mathrm{C}(D)$.
\begin{enumerate}
	\item $\minrank[D, q] = \minrank[D', q]$. 
	This holds since, for any choice of $c \in \<q\>^n$ and any $f \in \functions[D,q]$, the image of $f' \in \functions[D', q]$ defined by $f'_{V_0}(x) = c$ and $f'_{V_1}(x_{V_0}) = f(x_{V_0})$ has the same rank as $f$.

	\item $\minrank[D', q] = \minrank[D'', q]$. Let $v_1$ be redundant, $\NIn(S) = \NIn(v_1)$, and $f'' \in \functions[D'', q]$. Defining $f'(x)$ by $f'_u(x) = f''_u(x)$ for all $u \ne v_1$ and
	\[
		f'_{v_1}(x) = \begin{cases}
		0 & \text{if } \exists s \in S : f''_s(x) \ne 1\\
		1 & \text{otherwise},
		\end{cases} 
	\]
	then we see that $\rank(f') = \rank(f'')$.

	\item $\minrank[C,q] = \minrank[D'',q]$. Let $\NIn(u_0) = \NIn(v_0)$ and $\tilde{f} \in \functions[C, q]$. Let $f'' \in \functions[D'', q]$ be defined by $f''_v(x) = \tilde{f}(\tilde{x})$, where $\tilde{x}_v = x_v$ if $v \ne v_0$ and $\tilde{x}_{v_0} =  x_{u_0} \land x_{v_0}$, then once again $\rank(f'') = \rank(\tilde{f})$.
\end{enumerate}
\end{proof}

%The following are equivalent.
%\begin{enumerate}
%	\item $D$ is canonical, i.e. oriented-bipartite with no redundant vertices.
%
%	\item $D$ is the canonical version of itself.
%
%	\item $D$ is the canonical version of some other graph.
%\end{enumerate}

We now give bounds on $\minrank[C]$ for any canonical graph $C = (V,E).$ We denote the set of sources of $C$ as $A = \{ a_1, \dots, a_m \}$ and the set of sinks of $C$ as $B = \{b_1, \dots, b_n\}$. Let $L(C)$ be the maximum size of a sequence of vertices such that the in-neighbourhood of the $i$-th vertex is not contained in the in-neighbourhood of the previous vertices in the sequence, plus one:
\[
	L(C) = \max \{ k : \exists b_{j_1}, \dots, b_{j_k} \in B, \NIn(b_{j_l}) \not\subseteq \NIn(\{b_{j_1}, \dots, b_{j_{l-1}} \}) \forall 1 \le i \le k \} + 1.
\]
From $C = (A \cup B, E)$, construct the simple graph $G$ on $B$, where two vertices $b_j, b_{j'} \in B$ are adjacent if and only if $\NIn(b_j) \cap \NIn(b_{j'}) \ne \emptyset$. Let $U(C)$ denote the number of independent sets of $G$.

\begin{proposition} \label{prop:bounds_mr}
For any canonical graph $C$, 
\[
	L(C) \le \minrank[C] \le U(C).
\]
\end{proposition}

\begin{proof}
We first prove the lower bound. For any $S \subseteq B$, we denote $\r(S) := \minrank[ C[A \cup S] ]$.% as the minimum rank for the induced subgraph $C[A \cup S]$. 

\begin{claim}
For any $S \subseteq B$ and any $b \in B$ such that $\NIn(b) \not\subseteq \NIn(S)$, 
\[
	\r(S \cup b) \ge \r(S) + 1.
\] 
\end{claim}

\begin{proof}
Let $N = \NIn(S)$ and $P = \NIn(b) \setminus N$; by hypothesis, $P$ is not empty. Choose a function $f \in \functions[D,q]$ and $p \in P$ and let $y,z \in \<q\>^{|A|}$ illustrate the influence of the arc $pb$: $y_{-p} = z_{-p}$ and $f_b(y) \ne f_b(z)$. We then have
\[
	\Ima(f_{S \cup b}) \supseteq \{ ( f_S(x_N), f_b(x_N, y_P) ) : x_N \in \<q\>^{|N|} \} \cup \{ ( f_S(z_N), f_b(z_N, z_P) ) \}.
\]
The first set has size $\rank(f_S)$ and both sets are disjoint, thus $\rank(f_{S \cup b}) \ge \rank(f_S) + 1$. This proves the claim. 
\end{proof}

Since $\r(\emptyset) = 1$, applying the claim recursively with $S = \{b_{j_1}, \dots, b_{j_l}\}$ and $b = b_{j_{l+1}}$ for $0 \le l \le k-1$ yields the lower bound.

We now prove the upper bound. Consider $f \in \functions[C, n]$ given by
\begin{alignat*}{2}
	f_{a_i}(x) &:= 0 & \qquad & \forall a_i \in A,\\
	f_{b_j}(x) &:= \begin{cases}
	1 & \text{if } x_{\NIn(b_j)} = (j-1, \dots, j-1)\\
	0 & \text{otherwise}
	\end{cases} & \qquad & \forall b_j \in B.
\end{alignat*}
For any $y \in \{0,1\}^{m + n}$, its support is $\{v \in V : y_v = 1\}$. We prove that the image of $f$ consists of all $y \in \{0,1\}^{m + n}$ such that the support of $y$ is an independent set of $G$. Indeed, if the support of $y$, say $b_{j_1}, \dots, b_{j_k}$, is an independent set of $G$, then $y = f(x)$, with $x_{\NIn(b_{j_l})} = (l-1, \dots, l-1)$. Conversely, suppose the support of $y$ is not an independent set of $G$. If $y_a = 1$ for some $a \in A$, then clearly $y \notin \Ima(f)$. Otherwise, there exist $b_j, b_{j'} \in B$ such that $y_{b_j} = y_{b_{j'}} = 1$ and $\NIn(b_j) \cap \NIn(b_{j'}) \ne \emptyset$. If $y = f(x)$, then for any $a \in \NIn(b_j) \cap \NIn(b_{j'})$, we must have $x_a = j-1$ and $x_a = j' - 1$, which is a contradiction.
\end{proof}

We can characterise exactly when these bounds meet. Firstly, the transitive tournament with loops on $n$ vertices is $T_n = ([n], E)$ with $E = \{ ij : i \le j \}$. Let us define a family of graphs generalising $T_n$. Let $\mathcal{T}_n$ be the family of graphs $D = (L \cup R, E)$ such that: $|R| = n$; for any $r \in R$, there is a loop on $r$; and for any $r,r' \in R$, either $rr' \in E$ or there is $l \in L$ such that $lr, lr' \in E$. In particular, if $L$ is empty, then we obtain a graph isomorphic to $T_n$. We note that if $H \in \mathcal{T}_n$, then 
\[
	L(\mathrm{C}(H)) = \minrank[H,n] = \minrank[H] = U(\mathrm{C}(H)) = n+1.
\]

\begin{proposition} \label{prop:matching_bounds}
For a canonical graph $C = (A \cup B, E)$ with $|B| = n$, 
\[
	L(C) = \minrank[C] = U(C)
\]
if and only if $C$ is isomorphic to $\mathrm{C}(H)$ for some $H \in \mathcal{T}_n$.
\end{proposition}

\begin{proof}
We have $U(C) \ge n + 1$, with equality if and only if $\NIn(b) \cap \NIn(b') \ne \emptyset$ for all $b, b' \in B$. Conversely, $L(C) \le n + 1$, with equality if we can sort the vertices of $B$ so that $\NIn(b_j) \not\subseteq \NIn(\{b_1, \dots, b_{j-1}\})$ for all $j$. Thus, if $L(C) = U(C)$, both conditions must hold, and in particular $m \ge n$. Sort the vertices of $A$ such that $a_j \in \NIn(b_j) \setminus \NIn(\{b_1, \dots, b_{j-1}\})$ for all $1 \le j \le k$ and consider the graph $H$ on $L \cup R$ such that $R$ has a loop on each vertex, $r_ir_j \in E$ if and only if $a_ib_j$ is an arc in $C$ for all $1 \le i,j \le n$, and $l_ir_j$ is an arc if and only if $a_{n+i}b_j$ is an arc for all $1 \le i \le m-n$ and $1 \le j \le E$. Then $H \in \mathcal{T}_n$ and $\mathrm{C}(H) = C$.
\end{proof}

%\begin{corollary} \label{cor:Tn}
%For all $n$, $\minrank[T_n] = n+1$.
%\end{corollary}

We can refine the lower bound as follows. The function $\r(S)$ defined in the proof of Proposition \ref{prop:bounds_mr} satisfies not only $\r(S \cup b) \ge \r(S) + 1$ if $\NIn(b) \not\subseteq \NIn(S)$ and the initial condition $\r(\emptyset) = 1$, but also
\[
	\r(S \cup T) = \r(S) \cdot \r(T) \qquad \forall S,T \subseteq B, \NIn(S) \cap \NIn(T) = \emptyset.
\]
The proof of that equality is straightforward. We obtain that $\minrank[C, q] \ge L'(C)$, where $L'(C)$ is the optimal solution of the following minimisation problem, over all functions $r : 2^B \to \mathbb{N}$:
\begin{alignat*}{3}
	\min 			& \qquad & r(B) 		 \\
	\text{s.t. } 	& \qquad & r(\emptyset)	&= 1\\
					& \qquad & r(S \cup b) 	&\ge r(S) + 1 		& \qquad & \text{if } S \subseteq B, \NIn(b) \not\subseteq \NIn(S) \\
					& \qquad & r(S \cup T) 	&= r(S) \cdot r(T) 	& \qquad & \text{if } S,T \subseteq B, \NIn(S) \cap \NIn(T) = \emptyset.
\end{alignat*}

For any $D$, the conjunctive network on $D$ is $f \in \functions[D,2]$ such that for all $v \in [n]$, 
\[
	f_v(x) = \bigwedge_{u \in \NIn(v)} x_u,
\]
where an empty conjunction is equal to $1$. The rank of the conjunctive network on $D$ is denoted $\crank[D]$. By adapting the proof of Lemma \ref{lem:canonical}, we easily obtain $\crank[D] = \crank[\mathrm{C}(D)]$. We can show that the upper bound $U(C)$ is not always reached. For instance, let $C$ be as in Figure \ref{fig:tilde}. Then it is easy to verify that $U(C) = 8$, while $\crank[C] = 7$.

\begin{figure}
\centering
\begin{tikzpicture}[scale = 2]
	\begin{scope}[xshift = 0.5cm]
		\node[draw, shape=circle] (a1) at (1,1) {$a_1$};
		\node[draw, shape=circle] (a2) at (2,1) {$a_2$};
		\node[draw, shape=circle] (a3) at (3,1) {$a_3$};
%		\node[draw, shape=circle] (a4) at (4,1) {$a_4$};
	\end{scope}

	\node[draw, shape=circle] (b1) at (1,0) {$b_1$};
	\node[draw, shape=circle] (b2) at (2,0) {$b_2$};
	\node[draw, shape=circle] (b3) at (3,0) {$b_3$};
	\node[draw, shape=circle] (b4) at (4,0) {$b_4$};
%	\node[draw, shape=circle] (b5) at (5,0) {$b_5$};
	
	\draw[-latex] (a1) -- (b1);
	\draw[-latex] (a1) -- (b2);
	\draw[-latex] (a2) -- (b2);
	\draw[-latex] (a2) -- (b3);
	\draw[-latex] (a3) -- (b3);
	\draw[-latex] (a3) -- (b4);
%	\draw[-latex] (a4) -- (b4);
%	\draw[-latex] (a4) -- (b5);
\end{tikzpicture}
\caption{A canonical graph where $\crank[C] < U(C)$.} \label{fig:tilde}
\end{figure}

We now classify graphs with minimum rank $1$, $2$, or $2^n$. The first classification result is straightforward, but we include it as a template for the following results.

\begin{proposition} \label{prop:1}
For any graph $D$, the following are equivalent.
\begin{enumerate}[label=(\alph*)]
	\item $\crank[D] = 1$.

	\item $\minrank[D,2] = 1$.

	\item $\minrank[D] = 1$.

	\item $D$ is empty.

	\item $\mathrm{C}(D)$ is empty.
\end{enumerate}
\end{proposition}

\begin{theorem} \label{th:2}
For any graph $D$, the following are equivalent.
\begin{enumerate}[label=(\alph*)]
	\item \label{it:2c} $\crank[D] = 2$.

	\item \label{it:22} $\minrank[D,2] = 2$.

	\item \label{it:2q} $\minrank[D] = 2$.

	\item \label{it:2D} There exists a set $S$ such that for all $v$, either $v$ is a source or $\NIn(v) = S$. 

	\item \label{it:2t} $\mathrm{C}(D) \cong \vec{P}_1$.
\end{enumerate}
\end{theorem}

\begin{proof}
Clearly, \ref{it:2t} $\Rightarrow$ \ref{it:2c} $\Rightarrow$ \ref{it:22} $\Rightarrow$ \ref{it:2q}. We now prove \ref{it:2q} $\Rightarrow$ \ref{it:2D}. If $u,v$ are two vertices of $D$ such that $\NIn(u) \ne \NIn(v)$, $\NIn(u) \ne \emptyset$ and $\NIn(v) \ne \emptyset$, then $C$ will contain at least two sinks, corresponding to $u$ and $v$, respectively. Since these two sinks have different in-neighbourhoods in $C$, we obtain that $L(C) \ge 3$ and hence $\minrank[D] \ge 3$.

We finally prove \ref{it:2D} $\Rightarrow$ \ref{it:2t}. If \ref{it:2D} holds, then all but one sinks of $D'$ will be redundant. Removing them yields $D''$ with only one sink. Removing the redundant sources of $D''$ then yields $\mathrm{C}(D)$ with one source and one sink.
\end{proof}

\begin{theorem} \label{th:2n}
For any graph $D$ on $n$ vertices, the following are equivalent.
\begin{enumerate}[label=(\alph*)]
	\item \label{it:2nc} $\crank[D] = 2^n$.

	\item \label{it:2n2} $\minrank[D,2] = 2^n$.

	\item \label{it:2nq} $\minrank[D] = 2^n$.

	\item \label{it:2nD} $D$ is a disjoint union of cycles.

	\item \label{it:2nt} $\mathrm{C}(D)$ has $n$ connected components, all isomorphic to $\vec{P}_1$.
\end{enumerate}
\end{theorem}

\begin{proof}
Clearly, \ref{it:2nD} $\Rightarrow$ \ref{it:2nt}. We now prove that \ref{it:2nt} $\Rightarrow$ \ref{it:2nD}. Suppose \ref{it:2nt} holds, then the arcs of $C = \mathrm{C}(D)$ are $a_1b_1, \dots, a_nb_n$. Since $D$ has exactly $n$ vertices, we have $C = D'$ and all the arcs in $C$ correspond to arcs in $D$. Thus in $D$, every vertex has in- and out-degree equal to one, hence $D$ is the disjoint union of cycles.

Clearly, \ref{it:2nD} $\Rightarrow$ \ref{it:2nq} $\Rightarrow$ \ref{it:2n2} $\Rightarrow$ \ref{it:2nc}. We now prove that \ref{it:2nc} $\Rightarrow$ \ref{it:2nD}. If the conjunctive network on $D$ is a permutation of $\<2\>^n$, then all its local functions are balanced, i.e. $|f_v^{-1}(0)| = |f_v^{-1}(1)|$. Only a conjunction of one variable is balanced, thus all vertices have in-degree one in $D$. Moreover, since $\functions[D,2]$ contains a permutation of $\<2\>^n$, then all the vertices of $D$ must be covered by cycles \cite{Gad18}. Thus, $D$ is a disjoint union of cycles. 
\end{proof}

\subsection{The conjunctive network does not minimise the rank}

The results above seem to indicate that the conjunctive network has a low rank. Here we exhibit a graph $D$ such that the conjunctive network on $D$ does not minimise the rank over $\functions[D, 2]$, and the minimum rank over $\functions[D, n-1]$ is exponentially smaller than the minimum rank over $\functions[D, 2]$. For any $n \ge 2$, let $D = ([n+1], E) \in \mathcal{T}_n$ with
\[
	E = \{ 1v : v \ne 1 \} \cup \{ vv : v \ne 1 \}.
\]

\begin{theorem}
For all odd $n$, we have
\begin{align*}
	\crank[D]  &= 2^n,\\
	\minrank[D, 2] &= 2^{\lceil n/2 \rceil } + 2^{\lfloor n/2 \rfloor } - 1,\\
	\minrank[D, n] &= \minrank[D] = n+1.
\end{align*}
\end{theorem}

\begin{proof}
The image of the conjunctive network on $D$ is $\{ (1,y) : y \in \<2\>^n \}$, hence its rank is $2^n$. 

Now, let us give a lower bound on the minimum rank for any $f \in \functions[D, 2]$. For any $v \ge 2$, there exists at least one value of $x_1$ for which $f_v(x_1, x_v)$ depends on $x_v$:
$$
	X_v = \{ a \in \{0,1\} : f_v(a, 0) \ne f_v(a, 1) \}.
$$
Let $Z_0 = \{v : X_v = \{0\} \}$, $Z_1 = \{v : X_v = \{1\} \}$, $Z = \{v : X_v = \{0, 1\} \}$. We denote $x = (x_1, x_{Z_0}, x_{Z_1}, x_Z)$ for all $x \in \<2\>^n$. Then there exist $c^0 \in \{0,1\}^{|Z_0|}$ and $c^1 \in \{0,1\}^{|Z_1|}$ such that
\[
	\Ima( f ) = \{ x : x_1 = f_1, x_{Z_1} = c^1 \}  \cup \{ x : x_1 = f_1, x_{Z_0} = c^0 \}.
\]
The intersection between the two last sets is $\{ x : x_1 = f_1, x_{Z_0} = c^0, x_{Z_1} = c^1 \}$, hence the rank of $f$ is
$$
	\rank(f) = 2^{|Z_0| + |Z|} + 2^{|Z_1| + |Z|} - 2^{|Z|} \ge 2 \cdot 2^{\lceil n/2 \rceil } + 2^{\lfloor n/2 \rfloor } - 1.
$$

The lower bound is reached by $g \in \functions[D, 2]$, defined as follows:
\begin{alignat*}{2}
	g_1(x_1) &= 1, &&\\
	g_v(x_1, x_v) &= x_1 \land x_v 		&\qquad& 2 \le v \le \lceil n/2 \rceil + 1,\\
	g_v(x_1, x_v) &= \neg x_1 \land x_v &\qquad& \lceil n/2 \rceil + 2 \le v \le n+1.
\end{alignat*}
Then for $g$, $Z$ is empty, $|Z_0| = \lfloor n/2 \rfloor$, and $|Z_1| = \lceil n/2 \rceil$, and hence its rank reaches the lower bound.

Finally, since $D \in \mathcal{T}_n$, we have $\minrank[D, n] = \minrank[D] = n+1$.
\end{proof}

\section{Survey} \label{sec:survey}

We focus on three dynamical properties: the rank, the periodic rank and the number of fixed points. For each, we consider the minimum, the average and the maximum value it can take in $\functions[D,q]$ or $\functions(D,q)$. In order to highlight the influence of the interaction graph, we give those values for univariate functions in the set ${\bf F}_q := \{ \phi : [q] \to [q] \}$.

\subsection{Number of fixed points}

\paragraph{Minimum}

The minimum number of fixed points is settled. We have that $\minfix[D,q] = 1$ if $D$ is acyclic (an easy consequence of Robert's theorem), while $\minfix[D,q] = 0$ otherwise \cite{AS}.

In ${\bf F}_q$, the number of fixed-point free functions is exactly $(q-1)^q$, hence the proportion of fixed-point free functions tends to $1/e$. In fact, more is known: for every fixed $k$, the proportion of functions in ${\bf F}_q$ with exactly $k$ fixed points tends to $e^{-1}/k$. On the other hand, little is known about the proportion of fixed-point free functions in $\functions[D, q]$, when $q$ tends to infinity. In \cite{Gad17a}, it is shown that for $D = \vec{C}_n$, the directed cycle, that proportion does tend to $1/e$.

\begin{problem}
Is the proportion of fixed-point free functions in $\functions[D,q]$ positive, as $q$ tends to infinity?
\end{problem}

%\medskip
\paragraph{Average}

The average number of fixed points is equal to $1$ for every $D$ and $q$ \cite{Gad17a}.

Moreover, for $q \ge 3$, $\functions[D,q]$ contains a function with exactly one fixed point (see the section on the minimum number of periodic points). This is not true when $q=2$; for instance if $D = \vec{C}_n$, then every function in $\functions[\vec{C}_n, 2]$ has either two fixed points (a positive cycle) or no fixed points (a negative cycle).

\begin{problem}
Can we determine which graphs have a function with exactly one fixed point?
\end{problem}

%\medskip
\paragraph{Maximum}

%For instance, maximising the number of fixed points of an FDS based on its interaction graph was the subject of a lot of work, e.g. in \cite{Ara08,ADG04,GRR14,Ric09,Rii07}. The logarithm of the number of fixed points is notably upper bounded by the transversal number of its interaction graph \cite{ADG04,Rii07}. This upper bound is reached for large classes of graphs (e.g. perfect graphs) but is not tight in general \cite{Rii07}. Moreover, there is a dramatic change whether we assume that the FDS has an interaction graph equal to a certain digraph or only contained in that digraph (this is the distinction between guessing number and strict guessing number in \cite{GRF15}). 

The maximum number of fixed points in $\functions[D,q]$ or in $\functions(D,q)$ has been the subject of a lot of work. We shall split the main results into four parts: upper bounds, lower bounds, exact results and asymptotic values.

\medskip

Let $x,y \in \<q\>^n$ be two distinct fixed points of $f \in \functions[D, q]$, and consider the positions in which they differ: $\Delta(x, y) = \{ v \in [n] : x_v \ne y_v \}$. Then $D[ \Delta(x,y) ]$ contains a cycle \cite{GR11} (see also \cite{DR12} and \cite{GRR15}). From this observation, we can obtain two upper bounds on the number of fixed points. We denote the maximum cardinality of a code in $\<q\>^n$ with minimum distance $d$ as $A(n,q,d)$. The girth bound \cite{GR11} then asserts that 
\[
	\maxfix(D, q) \le A( n, q, \girth(D) ),
\]
while the feedback bound \cite{Ara08} gives
\[
	\maxfix(D, q) \le q^{\feedback(D)}.
\]
So far, we only know of graphs for which the feedback bound is tighter than the girth bound.

\begin{problem}
Is the girth bound always looser than the feedback bound?
\end{problem}

Riis developed another upper bound on $\maxfix(D, q)$ based on the entropy function and submodularity \cite{Rii07a}. He obtained that $\maxfix(D, q) \le q^{H(D)}$, where $H(D)$ is the optimal value of a solution of the following linear program (with one variable $h_S$ for any $S \subseteq [n]$):
\begin{alignat*}{3}
	\max &\qquad& h_V\\
	\text{s.t.} &\qquad& h_v 		&= 1	&\qquad& \forall v \in [n]\\
				&& 		h_{\NIn(v)} &= h_{\NIn(v) \cup v} &\qquad& \forall v \in [n]\\
				&&		h_S &\le h_T &\qquad&	\forall S \subseteq T \subseteq [n] \\
				&&		 h_{S \cap T} + h_{S \cup T}	&\le h_S + h_T &\qquad& \forall S,T \subseteq [n]
\end{alignat*}
In particular, if $D$ is an odd undirected cycle, i.e. $D = C_{2k+1}$ for $k \ge 2$, then $H(C_{2k+1}) = k + 1/2$ \cite{CM11}.

\medskip

Lower bounds for $\maxfix(D,q)$ can be obtained by packing smaller graphs with relatively high $\maxfix$. For instance, the complete graph $K_n$ satisfies $\maxfix(K_n, q) = q^{n - 1}$ for all $n$ and $q$ (a function with $q^{n-1}$ fixed points is $f_v(x) = - \sum_{u \ne v} x_u$). The clique partition number $\cliqueCover(D)$ is the minimum number of cliques required to cover all the vertices in the graph. Packing cliques then yields 
\[
	\maxfix(D,q) \ge q^{n - \cliqueCover(D)}.
\]
Moreover, the cycle $\vec{C}_n$ has $\maxfix(\vec{C}_n, q) = 1$ for all $n$ and $q$ (here, use $f_v(x) = x_{v-1}$); thus packing cycles yields 
\[
	\maxfix(D, q) \ge q^{\packing(D)}. 
\]
Here, $\packing(D)$ is the cycle packing number, i.e. the largest number of vertex-disjoint cycles in $D$.

These bounds can be refined by considering fractional strategies as follows. Since the alphabet $[q^k]$ is isomorphic to $[q]^k$, working in $\functions[D, q^k]$ can be viewed as working in $\functions[k \oplus D, q]$, where $k \oplus D = ([n] \times [k], \{ (u,i)(v,j) : uv \in E \})$. Therefore, $\maxfix(D, q^k) = \maxfix(k \oplus D, q)$. A fractional cycle packing of a graph $D$ is a family of cycles $B_1, \dots, B_s$ of $D$ together with non-negative weights $w_1, \dots, w_s$ such that $\sum_{i : v \in B_i} w_i \le 1$ for all $v$. In particular, a cycle packing is a fractional cycle packing where the weights belong to $\{0,1\}$. The maximum value of $\sum_{i=1}^s w_i$ over all fractional cycle packings is the fractional packing number and is denoted by $\fractionalPacking(D)$. Similarly, a fractional clique cover of a graph $D$ is a family of cliques $H_1, \dots, H_t$ of $D$ together with non-negative weights $w_1, \dots, w_t$ such that $\sum_{i : v \in H_i} w_i \ge 1$ for all $v$. Again, a clique cover is a fractional clique cover where all the weights belong to $\{0,1\}$. The minimum value of $\sum_{i=1}^t w_i$ over all fractional clique covers is the fractional clique cover number and is denoted by $\fractionalCliqueCover(D)$. For any $D$, there exist $k, k'$ such that 
\[
	\packing(k \oplus D) = k \fractionalPacking(D), \quad \cliqueCover(k' \oplus D) = k' \fractionalCliqueCover(D).
\] 
(See \cite{SU97} for more on fractional graph theory.) Therefore, there exist $k$ and $k'$ such that
\begin{align*}
	\maxfix(D, q^k) &\ge q^{k(n - \fractionalCliqueCover(D))}\\
	\maxfix(D, q^{k'}) &\ge q^{k' \fractionalPacking(D)}.
\end{align*}
In particular, if $D = C_{2k+1}$, then $n - \fractionalCliqueCover(C_{2k+1}) = k + 1/2$.

We finally note some lower bounds on $\maxfix(D,q)$ obtained from the guessing graph approach in \cite{GR11}. We have the code bound
\[
	\maxfix(D, q) \ge A(n, q, n - \delta(D) + 1),
\]
where $\delta(D)$ is the minimum in-degree of a vertex in $D$. In particular, $\maxfix(D,q) \ge q^{\delta(D)}$ for any prime power $q \ge n$. For other values of $q$, we have the competing bound
\[
	\maxfix(D, q) \ge \frac{ q^{\delta(D)} }{ n }.
\]

We now turn to lower bounds on $\maxfix[D,q]$. First of all, we can easily add ``ghost dependencies'' to any $f \in \functions(D,q-1)$ to create $f' \in \functions[D, q]$ with at least as many fixed points as $f$. Therefore, 
\[
	\maxfix[D,q] \ge \maxfix(D, q-1). 
\]
Remarkably, \cite{ARS17} shows two ways how packing cycles can also help with $\maxfix[D, 2]$. Firstly, 
\[
	\maxfix[D, 2] \ge \packing(D) + 1.
\]
The function achieving this number of fixed points can be schematically described as follows. Let $C_1, \dots, C_\nu$ be a collection of disjoint cycles; for the sake of simplicity, assume that they cover all the vertices of $D$. For any $0 \le i \le \nu$, let $S_i = C_1 \cup \dots \cup C_{i-1}$, $T_i = C_{i+1} \cup \dots \cup C_\nu$, and say that the predecessor of $v_i$ in $C_i$ is $u_i$. Then let
\[
	f_{v_i}(x) =   \left( x_{u_i} \land  \bigwedge_{s \in \NIn(v_i) \cap S_i} x_s   \right) \lor \bigvee_{t \in \NIn(v_i) \cap T_i} x_t.
\]
The fixed points of $f$ are all the $y^k \in [2]^n$ such that $y^k_{S_k} = (1, \dots, 1)$ and $y^k_{T_k} = (0, \dots, 0)$ for $0 \le k \le \nu$. This construction is monotone; the $\packing(D) + 1$ lower bound is proved to be tight for monotone Boolean networks in \cite{ARS17}, but it is still open whether it is tight in general.

\begin{problem}
Does there exist a graph $D$ with $\maxfix[D,2] =  \packing(D) + 1$?
\end{problem}

Secondly, \cite{ARS17} introduces the notion of principal packing. Intuitively, this is a packing of $k$ cycles where the effect of the arcs amongst the cycles can be compensated, and hence we recover exactly $2^k$ fixed points. Formally, given a cycle packing $C_1, \dots, C_k$, a path $P$ of $D$ is said principal if it has no arc and no internal vertex in the packing. A cycle $C$ is also said principal if it is not in the packing and has exactly one vertex in the packing. We then say that the packing is principal if the following holds: for every cycle $C_i$ and for every vertex $v$ in $C_i$, if there exists a principal path from a cycle $C_j\neq C_i$ to $v$, then either there exists a principal path from $C_i$ or a source to $v$, or there exists a special cycle containing $v$. We denote by $\packing'(D)$ the maximum size of a special packing in $D$. Then $\maxfix[D,2] \ge 2^{\packing'}$. A comparison between $\packing$ and $\packing'$ is also carried out in \cite{ARS17}.

\medskip

The value of $\maxfix(D,q)$ or $\maxfix[D,q]$ is known for many classes of graphs. First of all, if $\feedback(D) = \cliqueCover(D)$, then the clique cover strategy reaches the feedback bound; this is notably the case when $D$ is a symmetric perfect graph. Moreover, the entropy upper bound meets the fractional clique cover lower bound for odd cycles and their complements \cite{CM11}: $H(C_{2m+1}) = m + 1/2 = (2m+1) - \fractionalCliqueCover(C_{2m+1})$ and $H(\bar{C}_{2m+1}) = 2m - 1 - 1/k = (2m+1) - \fractionalCliqueCover(\bar{C}_{2m+1})$.

The value of $\maxfix(D,q)$ is equal to $q^{\feedback(D)}$ if $\feedback(D) \in \{0, 1, n\}$, since in those cases, $\feedback(D) = \packing(D)$. Moreover, it is implicit from the work in \cite{GR11, GRF16}, that this is still the case when $\feedback(D) = n-1$. We also have $\maxfix(D, q) = q^{\feedback(D)}$ for $\feedback(D) = 2$; this is a highly nontrivial result, presented in the context of linear network coding in \cite{SD10}. We also have $\maxfix[D,q] = q^{\feedback(D)}$ for $\feedback(D) \in \{0,1\}$ \cite{GRF16}. The value of $\maxfix[D,q]$ is also known for $\feedback(D) = n$, i.e. for loop-full graphs \cite{GRF16}. For any loopless $D$, an in-dominating set is a set of vertices $X \subseteq [n]$ such that for all $v \in [n]$ with positive in-degree, either $v \in X$ or $\NIn(v) \cap X \ne \emptyset$. Denote the number of in-dominating sets of $D$ of size $k$ by $I_k(D)$. Let $\loopfull{D}$ be the graph obtained by adding a loop on every vertex of $D$. Then
\[
	\maxfix[\loopfull{D}, q] = \sum_{k=0}^n (q-1)^k I_k(D).
\]

\medskip

The $q$-strict guessing number is $\guessing[D,q] := \log_q \maxfix [D,q]$, while the $q$-guessing number is $\guessing(D,q) := \log_q \maxfix (D,q)$. The guessing number and the strict guessing number tend to the same limit as $q$ tends to infinity \cite{CM11}:
\[
	\guessing(D) = \sup_{q \ge 2} \guessing[D,q] = \lim_{q \to \infty} \guessing[D,q] =  \sup_{q \ge 2} \guessing(D,q) = \lim_{q \to \infty} \guessing(D,q).
\]
Let us call this limit the asymptotic guessing number of $D$. Clearly, if $\feedback(D) = n$, then for all $q$ we have $\guessing[D,q] < \guessing(D,q) = \guessing(D) = n$, thus the asymptotic guessing number is not always achieved for by the strict $q$-guessing number.

\begin{problem}
Is the asymptotic guessing number always achieved by the guessing number, i.e. for any $D$, does there exist $q$ such that $\guessing(D,q) = \guessing(D)$?
\end{problem}

Another major problem about the asymptotic guessing number is the values it can take. Let $G$ be the set of asymptotic guessing numbers of all finite graphs, while $G'$ is the set of asymptotic guessing numbers of all finite symmetric graphs. $G'$ is sparse, since \cite{Gad18} shows that $G' \cap [0,k]$ is finite for all positive integers $k$.

\begin{problem}
Let $G$ be the set of all asymptotic guessing numbers, then is $G \cap [0,k]$ finite for all $k$?
\end{problem}

\begin{tabular}{| m{2cm} || m{2cm} | m{5cm} | m{5cm} |}
	\hline
	\textbf{Fixed points}	& ${\bf F}_q$				& $\functions[D,q]$															& $\functions(D,q)$\\
	\hline \hline
	Minimum						& 0 					& 1 if $D$ is acyclic \newline 0 otherwise			& 0\\
	\hline
	Average						& 1 					& 1 		& 1\\
	\hline
	Maximum					& $q$	
	& 	$= q^\feedback$ if $\feedback \in \{0,1\}$ 
		\newline
		$= \sum_k (q-1)^k I_k$ if $\feedback = n$ 
		\newline
		$\ge \maxfix(D,q-1)$ 
	 	\newline
		$\ge \packing + 1$ if $q=2$
		\newline
	 	$\ge 2^{\packing'}$ if $q = 2$ 
	&	$=q^\feedback$ if $\feedback \in \{0,1,2,n-1,n\}$ 
		\newline
		$\le A(n,q,\girth)$ 
		\newline
		$\le q^{\feedback}$ 
		\newline
		$\le q^H$ 
		\newline
		$\ge q^{n - \fractionalCliqueCover}$ 
		\newline
		$\ge q^{\fractionalPacking}$
		\newline
		$\ge A(n,q,n - \delta + 1)$
		\newline
		$\ge q^\delta/n$ \\
	\hline
\end{tabular}

\subsection{Number of periodic points}

\paragraph{Minimum} 
A function $f \in \functions(n,q)$ is nilpotent if it only has one periodic point. Equivalently, $f$ is nilpotent if there exists a state $y \in \<q\>^n$ such that $f^k(x) = y$ for all $x \in \<q\>^n$ and some $k \ge 1$. The smallest $k$ for which this holds is referred to as the class of $f$. For instance, $f$ is constant if and only if it is nilpotent of class $1$; Robert's theorem then asserts that any FDS with an acyclic interaction graph is nilpotent of class at most $n$. The study of the existence of nilpotent functions in $\functions[D,q]$ was initiated in \cite{GR16}; let us review some of the results therein.

First of all, for any $q \ge 3$ and any $D$, $\functions[D,q]$ contains a nilpotent function of class two, which is easy to describe:
\[
	f_v(x) = \begin{cases}
	0 & \text{if } x_u \in \{0,1\} \, \forall u \in \NIn(v)\\
	1 & \text{otherwise.}
	\end{cases}
\]
We then focus on the Boolean case, where the situation is much more complex. For instance, the cycle $\vec{C}_n$ only admits bijective Boolean networks, hence the minimum number of periodic points is $2^n$. Many graphs $D$ do admit nilpotent Boolean networks. Without loss, we can assume that $D$ is strong. Then $D$ admits a nilpotent Boolean network if one of the following holds:
\begin{enumerate}
	\item $D$ has a loop (and is not the graph with one vertex and a loop on it);

	\item $D$ is symmetric (and is not $K_2$);

	\item $D$ contains a primitive strict spanning subgraph.
\end{enumerate}
But in general, we do not know which graphs admit a nilpotent Boolean network.

\begin{problem}
Which graphs admit a nilpotent Boolean network? Is there a polynomial-time algorithm to determine whether a graph admits a nilpotent Boolean network?
\end{problem}

%\begin{problem}
%Does there exist a graph which admits a function with a single fixed point and yet does not admit any nilpotent function? yes. What about strong ones?
%\end{problem}

If $D$ admits a nilpotent Boolean network, then we can define the class of $D$ as the minimal class of a nilpotent Boolean network in $\functions[D,2]$. If $D$ falls into one of the three categories above, then its class is at most quadratic in $n$. Conversely, the cycle $\vec{C}_n$ with one added loop does admit a nilpotent Boolean network, but its class is $2n-1$. 

\begin{problem}
What is the maximum class of a nilpotent graph on $n$ vertices? Is it polynomial?
\end{problem}

%\medskip

\paragraph{Average} 
The average number of periodic points in ${\bf F}_q$ is found to be asymptotically $\sqrt{\pi q/2}$ in \cite{FO89}, where many other properties of random mappings are derived. This result is far from trivial, and its proof uses the machinery of analytic combinatorics described in detail in \cite{FS09}. This is the only result we know so far (omitting the obvious fact that the average number of periodic points is equal to $1$ if $D$ is acyclic), and it is likely that any advance in that area would require sophisticated and involved proofs.

%\medskip

\begin{problem}
Can we obtain any meaningful result on the average number of periodic points?
\end{problem}

\paragraph{Maximum}
The maximum number of periodic points was almost entirely determined in \cite{Gad17}. Let $\alpha_n(D)$ denote the maximum number of vertices of $D$ that can be covered by disjoint cycles. Then there exists a natural function $f \in \functions(D,q)$ with $q^{\alpha_n(D)}$ periodic points, described as follows. Fix a collection of disjoint cycles that altogether cover $\alpha_n(D)$ vertices; call the set of uncovered vertices $U$. Then, if $v'$ is the predecessor of $v$ on one of the cycles, then let $f_v(x) = x_{v'}$, and say $f_u(x) = 0$ for all $u \in U$. Clearly, any $x$ with $x_U = (0, \dots, 0)$ is a periodic point of $f$.

This simple strategy is in fact optimal: 
\[
	\maxper(D,q) = q^{\alpha_n(D)}
\]
for all $D$ and $q \ge 2$. For the case of $\functions[D,q]$, we can reach $q^{\alpha_n(D)}$ whenever $q \ge 3$; however, there are graphs $D$ such that $\maxper[D,2] < 2^{\alpha_n(D)}$ (for instance, the cycle $\vec{C}_n^\circ$ with a loop on each vertex, for $n \ge 2$). 

\begin{problem}
Which classes of graphs do not attain the $2^{\alpha_n(D)}$ upper bound on $\maxper[D,2]$?
\end{problem}

In general, the value of $\maxper[D,2]$ is open.

\begin{problem}
Can we derive a meaningful upper bound on $\maxper[D,2]$?
\end{problem}

\begin{tabular}{| m{2cm} || m{2cm} | m{5cm} | m{5cm} |}
	\hline
	\textbf{Periodic points}					& ${\bf F}_q$				& $\functions[D, q]$															& $\functions(D,q)$\\
	\hline \hline
	Minimum	& 1 	
	&	1 if $q \ge 3$ 
		\newline
		1 if $q=2$ for many graphs $D$
		\newline
		$2^n$ if $q=2$ and $D = \vec{C}_n$ 
	& 1\\
	\hline
	Average						& $\sim \sqrt{\pi q/2}$	& 1 if $D$ is acyclic			& 1 if $D$ is acyclic\\
	\hline
	Maximum	& $q$	& $q^{\alpha_n}$ if $q \ge 3$ & $q^{\alpha_n}$\\
	\hline
\end{tabular}

\subsection{Number of images}

\paragraph{Minimum} The minimum rank is studied above.

%\medskip

\paragraph{Average} In ${\bf F}_q$, a simple counting argument shows that the average rank is given by $(1 - (1 - q^{-1})^q )q$, which tends to $(1 - e^{-1})q$ for $q$ large. Similarly, the average rank in $\functions[D,q]$ is only a constant away from $\maxrank[D,q]$ (the latter being equal to $q^{\alpha_1(D)}$, see below) \cite{Gad17}. More precisely, for every $D$, there is a constant $c_D > 0$ such that $\avgrank[D,q] \ge c_D \cdot \maxrank[D,q]$ for all $q$ sufficiently large. 

\begin{problem}
For any $D$, does the ratio $\frac{ \avgrank[D,q] }{ \maxrank[D,q] }$ tend to a limit as $q$ tends to infinity?
\end{problem}

The ratio can decrease exponentially with $\alpha_1(D)$, since $\avgrank[D, q] \sim (1 - e^{-1})^n q^n$ if $D$ is the graph with only $n$ loops. Nonetheless, \cite{Gad17} only shows that $c_D \ge 2^{-3 (2^{\alpha_1(D)} - 1)}$, so there is arguably room for improvement. 

\begin{problem}
Does there exist an absolute constant $K > 0$ such that for any $D$, $\avgrank[D,q] > (Kq)^{\alpha_1(D)}$ for $q$ sufficiently large?
\end{problem}

%\medskip

\paragraph{Maximum} The maximum rank behaves in a very similar fashion to the maximum periodic rank, and was also almost completely determined in \cite{Gad17}. This time, say two arcs $uv$ and $u'v'$ are independent if $u \ne u'$ and $v \ne v'$. Let $\alpha_1(D)$ be the maximum number of pairwise independent arcs of $D$. Again, we can construct a natural function in $\functions(D,q)$ with $q^{\alpha_1(D)}$ images. Let $u_1v_1, \dots, u_kv_k$ be a family of pairwise independent arcs, denote $U = [n] \setminus \{v_1, \dots, v_k\}$ and let $f_{v_i}(x) = u_i$ for all $i$ and $f_u(x) = 0$ for $u \in U$; then the image of $f$ is all the $x$ such that $x_U = (0, \dots, 0)$. This strategy is optimal: for any $D$ and $q \ge 2$,
\[
	\maxrank(D, q) = q^{\alpha_1(D)}.
\]
And similarly to periodic points, the maximum rank in $\functions[D,q]$ is equal to $q^{\alpha_1(D)}$ whenever $q \ge 3$, but not necessarily when $q = 2$.

\begin{problem}
Which classes of graphs do not attain the $2^{\alpha_1(D)}$ upper bound on $\maxrank[D,2]$?
\end{problem}

In general, the value of $\maxrank[D,2]$ is also open.

\begin{problem}
Can we derive a meaningful upper bound on $\maxrank[D,2]$?
\end{problem}

\begin{tabular}{| m{2cm} || m{2cm} | m{5cm} | m{5cm} |}
	\hline
	\textbf{Images}					& $\functions(1,q)$				& $\functions[D, q]$															& $\functions(D,q)$\\
	\hline \hline
	Minimum						& 1 					& See above		& 1\\
	\hline
	Average						& $\sim (1 - e^{-1}) q$	& $\ge c_D q^{\alpha_1}$ & $\ge c_D q^{\alpha_1}$\\
	\hline
	Maximum	& $q$	& $q^{\alpha_1}$ if $q \ge 3$ & $q^{\alpha_1}$ \\
	\hline
\end{tabular}

%\bibliographystyle{amsplain}
%\bibliography{FDS}

\begin{thebibliography}{10}

\bibitem{ADG04a}
J.~Aracena, J.~Demongeot, and Eric Goles, \emph{Fixed points and maximal
  independent sets in {AND}-{OR} networks}, Discrete Applied Mathematics
  \textbf{138} (2004), 277--288.

\bibitem{ARS14}
J.~Aracena, A.~Richard, and L.~Salinas, \emph{Maximum number of fixed points in
  and-or-not networks}, Journal of Computer and System Sciences \textbf{80}
  (2014), no.~7, 1175 -- 1190.

\bibitem{Ara08}
J.~Aracena, \emph{Maximum number of fixed points in regulatory {B}oolean
  networks}, Bulletin of mathematical biology \textbf{70} (2008), 1398--1409.

\bibitem{ARS17}
J.~Aracena, A.~Richard, and L.~Salinas, \emph{Number of fixed points
  and disjoint cycles in monotone boolean networks}, SIAM journal on Discrete
  mathematics \textbf{31} (2017), 1702--1725.

\bibitem{AS}
J.~Aracena and L.~Salinas, \emph{Private communication}.

\bibitem{BG09a}
J.~Bang-Jensen and G.~Gutin, \emph{Digraphs: Theory, algorithms and
  applications}, Springer, 2009.

\bibitem{BCG18}
F.~Bridoux, A.~Castillo-Ramirez, and M.~Gadouleau,
  \emph{Complete simulation of automata networks}, available
  at \url{https://arxiv.org/abs/1504.00169}.

\bibitem{CM11}
D.~Christofides and K.~Markstr\"{o}m, \emph{The guessing number of
  undirected graphs}, Electronic Journal of Combinatorics \textbf{18} (2011),
  no.~1, 1--19.

\bibitem{DNS12}
J.~Demongeot, M.~Noual, and S.~Sen\'e, \emph{Combinatorics of {B}oolean
  automata circuits dynamics}, Discrete Applied Mathematics \textbf{160}
  (2012), 398--415.

\bibitem{DR12}
G.~Didier and E.~Remy, \emph{Relations between gene regulatory networks
  and cell dynamics in boolean models}, Discrete Applied Mathematics
  \textbf{160} (2012), 21472157.

\bibitem{FO89}
P.~Flajolet and A.~M.~Odlyzko, \emph{Random mapping statistics}, Research
  Report, INRIA \textbf{RR-1114} (1989), 1--26.

\bibitem{FS09}
P.~Flajolet and R.~Sedgewick, \emph{Analytic Combinatorics},
  Cambridge University Press, 2009.

\bibitem{Gad17}
M.~Gadouleau, \emph{Maximum rank and periodic rank of finite dynamical
  systems}, available at \url{https://arxiv.org/abs/1512.01448}, January 2017.

\bibitem{Gad17a}
\bysame, \emph{On the stability and instability of finite dynamical systems
  with prescribed interaction graphs}, available at \url{https://arxiv.org/abs/1709.02171}, September 2017.

\bibitem{Gad18}
\bysame, \emph{On the possible values of the entropy of undirected graphs},
  Journal of Graph Theory \textbf{88} (2018), 302--311.

\bibitem{GR16}
M.~Gadouleau and A.~Richard, \emph{Simple dynamics on graphs},
  Theoretical Computer Science \textbf{628} (2016), 6277.

\bibitem{GRF16}
M.~Gadouleau, A.~Richard, and E.~Fanchon, \emph{Reduction and
  fixed points of boolean networks and linear network coding solvability}, IEEE
  Transactions on Information Theory \textbf{62} (2016), no.~5, 2504--2519.

\bibitem{GRR15}
M.~Gadouleau, A.~Richard, and S.~Riis, \emph{Fixed points of
  boolean networks, guessing graphs, and coding theory}, SIAM Journal on
  Discrete Mathematics \textbf{29} (2015), no.~4, 2312--2335.

\bibitem{GR11}
M.~Gadouleau and S.~Riis, \emph{Graph-theoretical constructions
  for graph entropy and network coding based communications}, IEEE Transactions
  on Information Theory \textbf{57} (2011), no.~10, 6703--6717.

\bibitem{GM12}
E.~Goles and M.~Noual, \emph{Disjunctive networks and update schedules},
  Advances in Applied Mathematics \textbf{48} (2012), no.~5, 646--662.

\bibitem{Gol85}
E.~Goles, \emph{Dynamics of positive automata networks}, Theoretical Computer
  Science \textbf{41} (1985), 19--32.

\bibitem{GT83}
E.~Goles and M.~Tchuente, \emph{Iterative behaviour of generalized majority
  functions}, Mathematical Social Sciences \textbf{4} (1983), 197--204.

\bibitem{NS17}
M.~Noual and S.~Sen{\'e}, \emph{Synchronism versus asynchronism in
  monotonic boolean automata networks}, Natural Computing (2017).

\bibitem{PR10}
L.~Paulev\'e and A.~Richard, \emph{Topological fixed points in boolean
  networks}, Comptes Rendus de l'Acad\'emie des Sciences - Series I -
  Mathematics \textbf{348} (2010), 825--828.

\bibitem{Rii06}
S.~Riis, \emph{Utilising public information in network coding}, General
  Theory of Information Transfer and Combinatorics, Lecture Notes in Computer
  Science, vol. 4123/2006, Springer, 2006, pp.~866--897.

\bibitem{Rii07a}
\bysame, \emph{Graph entropy, network coding and guessing games}, available at
  http://arxiv.org/abs/0711.4175, November 2007.

\bibitem{Rii07}
\bysame, \emph{Information flows, graphs and their guessing numbers}, The
  Electronic Journal of Combinatorics \textbf{14} (2007), 1--17.

\bibitem{Rob80}
F.~Robert, \emph{Iterations sur des ensembles finis et automates cellulaires
  contractants}, Linear Algebra and its Applications \textbf{29} (1980),
  393--412.

\bibitem{SU97}
E.~R. Scheinerman and Daniel~H. Ullman, \emph{Fractional graph theory},
  Wiley, 1997.

\bibitem{SD10}
S.~Shenvi and B.~K. Dey, \emph{A simple necessary and sufficient
  condition for the double unicast problem}, proc. ICC2010, 2010.

\end{thebibliography}

\providecommand{\bysame}{\leavevmode\hbox to3em{\hrulefill}\thinspace}
\providecommand{\MR}{\relax\ifhmode\unskip\space\fi MR }
\providecommand{\MRhref}[2]{%
  \href{http://www.ams.org/mathscinet-getitem?mr=#1}{#2}
}
\providecommand{\href}[2]{#2}

\end{document}